\documentclass{article}

\usepackage{fullpage}
\usepackage{amsmath,amsthm,amssymb,bm}

\usepackage{hyperref,url,xcolor}

\makeatletter
\newtheorem*{rep@theorem}{\rep@title}
\newcommand{\newreptheorem}[2]{%
\newenvironment{rep#1}[1]{%
 \def\rep@title{#2 \ref{##1}}%
 \begin{rep@theorem}}%
 {\end{rep@theorem}}}
\makeatother

\newcommand{\change}[1]{{#1}}

\newcommand{\PF}{\mathbb{P}\mathbb{F}}
\newcommand{\MDS}{\operatorname{MDS}}

\newcommand{\LDMDS}{\operatorname{LD-MDS}}
\newcommand{\wt}{\operatorname{wt}}
\newcommand{\supp}{\operatorname{supp}}
\newcommand{\eps}{\epsilon}

\newcommand{\F}{\mathbb F}

\newcommand{\Span}{\operatorname{span}}
\newcommand{\poly}{\operatorname{poly}}

\newtheorem{theorem}{Theorem}[section]
\newreptheorem{theorem}{Theorem}
\newtheorem{corollary}[theorem]{Corollary}
\newtheorem{lemma}[theorem]{Lemma}
\newtheorem{proposition}[theorem]{Proposition}

\newtheorem{claim}[theorem]{Claim}
\newtheorem{question}[theorem]{Question}

\theoremstyle{definition}
\newtheorem{definition}[theorem]{Definition}

\theoremstyle{remark}
\newtheorem{remark}[theorem]{Remark}

\hypersetup{
	colorlinks=true,
	linkcolor=blue,%
	citecolor=blue,
	linktoc=page
}

\title{Improved Field Size Bounds for Higher Order MDS Codes\footnote{Abridged version appeared in ISIT 2023.}}
\author{Joshua Brakensiek\thanks{Department of Electrial Engineering and Computer Science, University of California, Berkeley. Email: \texttt{josh.brakensiek@berkeley.edu}. Research conducted while a graduate student at Stanford University. Supported in part by a Microsoft Research PhD Fellowship.} \and
Manik Dhar\thanks{Department of Mathematics, Massachusetts Institute of Technology. Email: \texttt{dmanik@mit.edu}. Part of this work was done while this author was interning at Microsoft Research, and was a graduate student at Princeton University where his research supported by NSF grant DMS-1953807.} \and
Sivakanth Gopi\thanks{Microsoft Research. Email: \texttt{sigopi@microsoft.com}}}
\date{}

\begin{document}
\maketitle

\abstract{
  Higher order MDS codes are an interesting generalization of MDS codes recently introduced by \change{Brakensiek, Gopi, and Makam}~\cite{bgm2021mds}. In later works, they were shown to be intimately connected to optimally list-decodable codes and maximally recoverable tensor codes. Therefore (explicit) constructions of higher order MDS codes over small fields is an important open problem. Higher order MDS codes are denoted by $\MDS(\ell)$ where $\ell$ denotes the order of generality, $\MDS(2)$ codes are equivalent to the usual MDS codes. The best prior lower bound on the field size of an $\change{[}n,k\change{]}$-$\MDS(\ell)$ codes is $\Omega_\ell(n^{\ell-1})$, whereas the best known (non-explicit) upper bound is $O_\ell(n^{k(\ell-1)})$ which is exponential in the dimension.

  In this work, we nearly close this exponential gap between upper and lower bounds. We show that an $\change{[}n,k\change{]}$-$\MDS(3)$ codes requires a field of size $\Omega_k(n^{k-1})$, which is close to the known upper bound. Using the connection between higher order MDS codes and optimally list-decodable codes, we show that even for a list size of 2, a code which meets the optimal list-decoding Singleton bound requires exponential field size; this resolves an open question \change{by Shangguan and Tamo}~\cite{shangguan2020combinatorial,ST23}.

   We also give explicit constructions of $\change{[}n,k\change{]}$-$\MDS(\ell)$ code over fields of size $n^{(\ell k)^{O(\ell k)}}$. The smallest non-trivial case where we still do not have optimal constructions is $\change{[}n,3\change{]}$-$\MDS(3)$. In this case, the known lower bound on the field size is $\Omega(n^2)$ and the best known upper bounds are $O(n^5)$ for a non-explicit construction and $O(n^{32})$ for an explicit construction. In this paper, we give an explicit construction over fields of size $O(n^3)$ which comes very close to being optimal.
}
\newpage
\tableofcontents
\newpage

\section{Introduction}
The Singleton bounds states that the minimum distacnce of an $\change{[}n,k\change{]}$-code is at most $n-k+1$ \cite{Singleton1964maximum}.\footnote{This bound holds for non-linear codes as well, but in this paper we will only focus on linear codes defined over some finite field $\F.$ A (linear) $\change{[}n,k\change{]}$-code over $\F$ is a $k$-dimensional subspace of $\F^n$.} Codes which achieve this bound are called Maximum Distance Separable (MDS) codes. Reed-Solomon codes are a beautiful construction of MDS codes over fields of size just $O(n).$ Field size plays an important role in several applications of MDS codes. In distributed storage where MDS codes are extensively used, field size is the main determinant in the efficiency of encoding the data and recovering from failures \cite{huang2012erasure,Plank2013ScreamingFG}. Because of their distance optimality (MDS) and the small field size, Reed-Solomon codes are one of the most widely used codes both in practice and in theory.

In a recent paper,\footnote{The timeline presented in this paper is
  based on the initial arXiv posting dates of each paper. In
  particular, such chronological order is
  \cite{bgm2021mds,roth2021higher,bgm2022}.} \change{Brakensiek, Gopi, and Makam}~\cite{bgm2021mds}
introduced a generalization of MDS codes called \emph{higher order MDS
  codes}. \change{For a matrix $V \in \F^{k \times n}$ and subset $A \subseteq [n]$,
  we let $V_{A}$ be the span of the columns of $V$ indexed by $A$.}

\begin{definition}[Higher order MDS codes \change{(Definition~1.3~\cite{bgm2021mds})}]
  Let $C$ be an $\change{[}n,k\change{]}$-code with generator matrix $G$.  Let $\ell$ be
  a positive integer. We say that $C$ is $\MDS(\ell)$ if for any
  $\ell$ subsets $A_1, \hdots, A_\ell \subseteq [n]$ of size of at
  most $k$, we have that
  \begin{align}
    \dim (G_{A_1} \cap \cdots \cap G_{A_\ell}) = \dim (W_{A_1} \cap
    \cdots \cap W_{A_\ell}), \label{eq:MDS-L}
  \end{align}
  where $W_{k \times n}$ is a generic\footnote{One can think of the entries of $W$ as symbolic variables or randomly chosen from a sufficiently large field (c.f., \cite{bgm2021mds}).} matrix over the same field characteristic.\footnote{Note that $\MDS(\ell)$ is a property of the code $C$ and not a particular generator matrix $G$ used to generate $C$. This is because if $G$ satisfies (\ref{eq:MDS-L}) then $MG$ also satisfies (\ref{eq:MDS-L}) for any $k\times k$ invertible matrix $M$.}
\end{definition}

\change{
\begin{remark}
  By Lemma~3.1 of \cite{bgm2021mds}, it suffices to check (\ref{eq:MDS-L}) in the case that $|A_1| + \cdots + |A_\ell| = (\ell-1)k$. We often make this assumption in our proofs.
\end{remark}
}

\change{We say that a matrix $k \times n$ matrix $G$ is $[n,k]$-$\MDS(\ell)$ if $G$ is the generator matrix of a $[n,k]$-$\MDS(\ell)$ code.} Note that $\MDS(1)$ and $\MDS(2)$ coincide with $\MDS$~\cite{bgm2021mds}. For more nontrivial example, the columns of a generator matrix of an $\change{[}n,3\change{]}$-$\MDS(3)$ code form $n$ points in the projective plane $\PF^2$ such that no three points are collinear and additionally, no three lines obtained by joining disjoint pairs of points are concurrent. This is in contrast to an ordinary $\change{[}n,3\change{]}$-MDS code where we only require that no three points are collinear.

\change{Brakensiek, Gopi, and Makam}~\cite{bgm2021mds} \change{showed} that higher order MDS codes are equivalent to maximally recoverable (MR) tensor codes which were first introduced \change{by Gopalan, et.al.} \cite{gopalan2016}. A subsequent work \change{by Brakensiek, Gopi, and Makam}~\cite{bgm2022} showed that higher order MDS codes are equivalent to optimally list-decodable codes. We look at these two equivalences in turn.

A code $C$ is a \emph{$(m,n,a,b)$-tensor code}\change{, or \emph{product code}, or \emph{turbo product code}} if it can be expressed as \change{the tensor product of linear codes} $C_{col} \otimes C_{row}$, where $C_{col}$ is a $\change{[}m,m-a\change{]}$-code and
$C_{row}$ is a $\change{[}n,n-b\change{]}$-code. In other words, the codewords of $C$
are $m\times n$ matrices where each row belongs to $C_{row}$ and each
column belongs to $C_{col}.$ There are `$a$' parity checks per column
and `$b$' parity checks per row. Such a code $C$ is \emph{maximally
  recoverable} (abbreviated as MR) if it can recover from every erasure pattern $E \subseteq [m] \times [n]$
which can be recovered from by choosing a generic $C_{col}$ and $C_{row}$. Thus MR tensor codes are optimal codes since they can recover from any erasure pattern that is information theoretically possible to recover from. 
\change{Brakensiek, Gopi, and Makam}~\cite{bgm2021mds} defined $\MDS(\ell)$ codes \change{based on} the following proposition.
\begin{proposition}[Higher order MDS codes are equivalent to MR tensor codes~\cite{bgm2021mds}]
\label{prop:mrtc_mdsell}
Let $C=C_{col}\otimes C_{row}$ be an $(m,n,a=1,b)$-tensor code. Here $a = 1$ and thus $C_{col}$ is a parity check code. Then $C$ is maximally recoverable if and only if $C_{row}$ is $\MDS(m)$.  
\end{proposition}

A generalization of the Singleton bound was recently proved for list-decoding in \cite{shangguan2020combinatorial,roth2021higher,goldberg2021Singleton}. If an $\change{[}n,k\change{]}$-code is $(L,\rho)$-list-decodable\footnote{I.e., there are at most $L$ codewords in any Hamming ball of radius $\rho n.$ In contrast, average-radius list decoding says that from any point there are at most $L$ codewords with average Hamming distance of $\rho n$ from the point.}, then 
\begin{equation}
  \label{eq:LD_Singleton}
\rho \le \frac{L}{L+1}\left( 1-\frac{k}{n}\right).  
\end{equation}
Note that when $L=1,$ this reduces to the usual Singleton bound.
 Roth~\cite{roth2021higher} defined a higher order generalization of MDS codes as codes achieving this generalized Singleton bound for average-radius list-decoding.

\begin{definition}[List decodable-MDS codes~\cite{roth2021higher}]\label{def:ldmds}
	Let $C$ be a $\change{[}n,k\change{]}$-code. We say that $C$ is list decodable-$\MDS(L)$, denoted as $\LDMDS(L)$, if $C$ is $(L,\rho)$-average-radius list-decodable for $\rho=\frac{L}{L+1}\left( 1-\frac{k}{n}\right).$
 In other words, for any $y\in \F^n$, there does not exist $L+1$ \emph{distinct} codewords $c_0,c_1,\dots,c_L \in C$ such that 
	\begin{equation*}
		\sum_{i=0}^L \wt(c_i-y) \le (L+1)\rho n= L(n-k).
	\end{equation*}
	We say that $C$ is list $\LDMDS(\le L)$ if it is $\LDMDS(\ell)$ for all $1\le \ell \le L.$
\end{definition}

An equivalent way to define $\LDMDS(L)$ codes is using the parity check matrix $H_{(n-k)\times n}$ matrix of $C$ \cite{roth2021higher}. $C$ is $\LDMDS(L)$ if there doesn't exist $L+1$ \emph{distinct} vectors $e_0,e_1,\dots,e_L \in \F^n$ such that 
\begin{align*}
	\sum_{i=0}^L \wt(e_i) \le L(n-k) \text{ and }He_0=He_1=\dots=He_L.
\end{align*} 

\change{Note that the usual MDS codes are $\LDMDS(1).$} The list-decoding guarantees of $\LDMDS(L)$ are very strong. In particular, \change{$[n,k]$-}$\LDMDS(L)$ codes \change{approach \emph{list-decoding capacity}.} \change{Formally, we say that an $[n,k]$-code is $\eps$-close to list-decoding capacity if for all $y \in \F^n$, there exist at most $O_{\eps}(1)$ codewords within Hamming distance of $(1-\eps)n - k$ of $y$.} \change{Therefore, Definition~\ref{def:ldmds} implies that $[n,k]$-$\LDMDS(L)$ codes} get $\eps$-close to list-decoding capacity when $L\ge \change{\frac{n-k-\eps}{\eps n}}$. \change{Brakensiek, Gopi, and Makam}~\cite{bgm2022} \change{show} the equivalence between $\LDMDS$ codes and the dual of higher order MDS codes.

\begin{proposition}[$\LDMDS$ codes are the dual of higher order MDS codes~\cite{bgm2022}]\label{prop-listdec}
    If $C$ is a linear code then for all $\ell \ge 1$, $C$ is $\MDS(\ell+1)$ if and only if $C^\perp$ is $\LDMDS(\le \ell)$.
\end{proposition}
Besides the connection to MR tensor codes and optimally list-decodable codes, higher order MDS codes are also shown to be intimately related to MDS codes whose generator matrices are constrained to have a specific support and the GM-MDS conjecture~\cite{bgm2022}. Such matrices have many applications in coding theory, see~\cite{dau2014gmmds}.

\subsection{Our Results}

Our main result exponentially improves the lower bound for higher order MDS codes.

\begin{theorem}\label{thm:exp-mds-lb}
Let $C$ be an $\change{[}n,k\change{]}$-code over the field $\F$ which is $\MDS(3)$. Then, $|\F| \ge \binom{n-2}{k-1}-1$.
\end{theorem}

By Proposition~\ref{prop:mrtc_mdsell}, we have the following corollary for MR tensor codes.

\begin{corollary}
 Let $C$ be an $(m,n,1,b)$-tensor code over the field $\F$ then  $|\F|\ge \binom{n-2}{b-1}-1$.   
\end{corollary}

Proposition~\ref{prop-listdec}, implies the following corollary for $\LDMDS$ codes.
\begin{corollary}\label{cor:ldmds}
   Let $C$ be an $\change{[}n,k\change{]}$-code over the field $\F$ which is $\LDMDS(\le 2)$ then $\F\ge \binom{n-2}{k-1}-1$ 
\end{corollary}

In particular, we see that if $C$ is of constant rate and $\LDMDS(\le 2)$ then we would need exponential in $n$ field size. For applications to list decoding, this is only talking about
information-theoretically-optimal\footnote{Here we mean that the
  tradeoff between list decoding radius and list size is
  \emph{exactly} as specified by the generalized Singleton bound.}
average-radius list decoding.

\change{As worst-case list decoding is a weaker notion than average-radius list decoding, the lower bound of Corollary~\ref{cor:ldmds} does not immediately apply. However, we show that an exponential lower bound can also be proven in the worst-case setting.}

\begin{theorem}\label{thm:ldlb}
Let $n \ge k \ge 0$ be such that $n-k$ is divisible by $3$. Let $C$ be
an $\change{[}n,k\change{]}$-MDS code which is $(2, \frac{2(n-k)}{3n})$-worst-case list decodable, i.e., it matches the list-decoding Singleton bound (\ref{eq:LD_Singleton}) for $L=2$. Then, $C$ requires field
size $\binom{(n+2k)/3}{k-1} - 1.$
\end{theorem}

\begin{remark}
  This lower bound does \emph{not} say that error-correcting codes
  achieving list decoding capacity require exponential field size
  (which contradicts the known list decoding capacity of random linear
  codes~\cite{Guruswami2010OnTL}). Instead is just says that codes over subexponential fields
  cannot achieve capacity with the \emph{exactly optimal} list size. See Section~\ref{sec:sub} for a discussion on how this bound has been circumvented.
\end{remark}

\begin{remark}
  Also note that this lower bound only applies to MDS codes, as non-MDS codes already have a suboptimal tradeoff for list-size $1$.
\end{remark}

\subsubsection{Explicit constructions}

Our first explicit construction is a general construction for $\change{[}n,k\change{]}-\MDS(\ell)$ codes field size construction. Prior to this work no constructions for general values of $n,k,$ and $l$ were known.

\begin{theorem}\label{thm-con-gen}
There is an explicit $\change{[}n,k\change{]}$-$\MDS(\ell)$-code over field size
$n^{(\ell k)^{O(\ell k)}}$.
\end{theorem}

We now give a high level overview of our construction. We first convert the generic intersection condition~(\ref{eq:MDS-L}) into a determinant condition using Lemma~\ref{lem:mds5}. 
Let $p \ge n+k-1$ be a prime power and consider the field extension $\F_p[\alpha_1,\hdots,\alpha_{k\ell}]$ over the base field $\F_p$ where each $\alpha_i$ has an extension degree $D=\ell k^2$ over $\F_p[\alpha_1,\dots,\alpha_{i-1}]$. In particular, this means that any non-zero polynomial $p(x_1,x_2,\dots,x_{\ell k})$ with $\F_p$ coefficients and individual degree at most $D-1$, cannot vanish at $(\alpha_1,\alpha_2,\dots,\alpha_{\ell k})$. Similar ideas were also used in the doubly exponential $\MDS(3)$ construction of \change{Roth~\cite{roth2021higher} and Shangguan-Tamo~\cite{shangguan2020combinatorial}}.

But here, we depart from previous constructions. Our $\change{[}n,k\change{]}$-$\MDS(\ell)$ code $C$ is the Reed-Solomon code generated by $n$ $\F_p$-linear combinations of $\alpha_1,\dots,\alpha_{k\ell}$ such that any $\ell k$ generators are linearly independent over $\F_p$, this can be achieved using a Reed-Solomon code over $\F_p.$ The key step is then to show that the determinant obtained from Lemma~\ref{lem:mds5} is a non-zero polynomial in $\alpha_1,\alpha_2,\dots,\alpha_{\ell k}$ (it is easy to see that its individual degree is at most $D-1$). Here we crucially use the GM-MDS theorem~\cite{lovett2018gmmds,yildiz2019gmmds} (which requires $p\ge n+k-1$) and the $\F_p$ linear independence of our evaluation points to show that there is an $\F_p$ substitution to $\alpha_1,\dots,\alpha_{\ell k}$ which makes the determinant non-zero. This shows that the determinant polynomial is indeed non-zero.

Our next construction is in the specific case of $\change{[}n,3\change{]}$-$\MDS(3)$ codes which is the smallest non-trivial case of a higher-order MDS code.

\begin{theorem}\label{thm-con-three}
There exists an explicit $\change{[}n,3\change{]}$-$\MDS(3)$ code with field size $O(n^3)$.    
\end{theorem}

This improves on the earlier explicit construction of size
$O(n^{32})$~\cite{roth2021higher} and is only a factor $n$ away from
the current best lower bound of $\Omega(n^2)$~\cite{bgm2021mds}. We also
construct some explicit codes for $k = 4$ and $k = 5$.

\begin{theorem}\label{thm-con-four}
There exists an explicit $\change{[}n,4\change{]}$-$\MDS(3)$ code with field size $O(n^7)$.    
\end{theorem}

\begin{theorem}%
There exists an explicit $\change{[}n,5\change{]}$-$\MDS(3)$ code with field size $O(n^{50})$.    
\end{theorem}

As the dual of $\MDS(3)$ is also $\MDS(3)$~\cite{roth2021higher,bgm2021mds} the three constructions above also give us constructions for $\change{[}n,n-k\change{]}$-$\MDS(3)$ for $k=3,4,5$. The above constructions involve carefully analyzing the algebraic conditions Reed-Solomon codes need to satisfy to have the higher order MDS property and carefully selecting the evaluation points so that we can argue for their correctness directly or in some cases reduce to a simple check which can be performed by a computer program.

\subsection{Comparison with prior work}
\renewcommand{\arraystretch}{1.5}
\begin{table}[!ht]
\centering
\begin{tabular}{ |c|c|c|c|}
\hline
\change{$\MDS(\ell)$ Code}& \change{Field Size} & \change{Explicit?} & \change{References}\\
\hline
$\change{[}n,k\change{]}$-$\MDS(3)$ & $2^{k^{n}}$ & \change{Yes} &\cite{shangguan2020combinatorial}\\
\hline
$\change{[}n,k\change{]}$-$\MDS(4)$  &$2^{(3k)^{n}}$ & \change{Yes} & \cite{shangguan2020combinatorial} \\
\hline
$\change{[}n,k\change{]}$-$\MDS(3)$ &$n^{k^{2k}}$ & \change{Yes} & \cite{roth2021higher} \\
\hline
$\change{[}n,k\change{]}$-$\MDS(\ell)$  &$n^{O(\min\{k,n-k\}(\ell-1))}$ & \change{No} & \cite{bgm2021mds,kong2021new,bgm2022}\\
\hline
$\change{[}n,k\change{]}$-$\MDS(\ell)$ &$n^{(lk)^{O(lk)}}$& \change{Yes} & Theorem~\ref{thm-con-gen}\\
\hline
\hline
$\change{[}n,3\change{]}$-$\MDS(3)$  &$\change{O(n^{32})}$ & \change{Yes} & \cite{roth2021higher}\\
\hline
$\change{[}n,3\change{]}$-$\MDS(3)$ &$\change{O(n^{5})}$ & \change{No} & \change{\cite{kong2021new,roth2021higher}}\\
\hline
\change{$\change{[}n,n-2\change{]}$-$\MDS(4)$} & \change{$O(n^5)$} & \change{No} & \cite{kong2021new}\\
\hline
$\change{[}n,3\change{]}$-$\MDS(3)$&$\change{O(n^3)}$ & \change{Yes} & Theorem~\ref{thm-con-three}\\
\hline
$\change{[}n,4\change{]}$-$\MDS(3)$  &$\change{O(n^7)}$ & \change{Yes} & Theorem~\ref{thm-con-four}\\
\hline
$\change{[}n,5\change{]}$-$\MDS(3)$ &$\change{O(n^{50})}$ & \change{Yes} & Theorem~\ref{thm-con-four}\\
\hline
\end{tabular}
\caption{Table showing \change{known} upper bounds for the field size of $\change{[}n,k\change{]}$-$\MDS(\ell)$ codes.}\label{tbl:up}
\end{table}
\renewcommand{\arraystretch}{1}

As \change{Table~\ref{tbl:up}} shows, prior to our work no general explicit constructions were known. As mentioned earlier, as the dual of $\MDS(3)$ is also $\MDS(3)$~\cite{roth2021higher,bgm2021mds} any $\change{[}n,k\change{]}$-$\MDS(3)$ construction in the table above gives us a $\change{[}n,n-k\change{]}$-$\MDS(3)$ construction as well. \change{In the MR tensor code literature, no explicit upper bounds on the field size were previously known (beyond cases equivalent to MDS codes). Note that by Proposition~\ref{prop:mrtc_mdsell}, constructing $[n,k]$-$\MDS(m)$ codes is equivalent to constructing $(m,n,a=1,b=n-k)$-MR tensor codes. Some non-explicit constructions and lower bounds for MR tensor codes for small values of $a,b$ are given in \cite{kong2021new}, which we translated into the language of $\MDS(\ell)$ codes in Tables~\ref{tbl:up} and \ref{tbl:lo}. For example, our explicit $[n,3]$-$\MDS(3)$ code gives an $(m=3,n,a=1,b=3)$-MR tensor code over fields of size $O(n^3)$, which improves over the non-explicit $O(n^5)$ field size construction given in \cite{kong2021new}.}  We also see that even after our work there is an exponential gap in the exponent of the field size between the explicit and non-explicit constructions.

\renewcommand{\arraystretch}{1.5}
\begin{table}[!ht]
\centering
\begin{tabular}{ |c|c|c|c|}
\hline
\change{$\MDS(\ell)$ Code} & \change{Field Size} & \change{Reference}\\
\hline
$\change{[}n,n-2\change{]}$-$\MDS(4)$& $\Omega(n^2)$ & \cite{kong2021new} \\
\hline
\change{$\change{[}n,3\change{]}$-$\MDS(3)$} & \change{$\Omega(n)$} & \change{\cite{kong2021new}} \\
\hline
$\change{[}n,k\change{]}$-$\MDS(\ell)$&$\Omega_{\ell,k}(n^{\min\{\ell,k,n-k\}-1})$& \cite{bgm2021mds}\\
\hline
$\change{[}n,k\change{]}$-$\MDS(3)$ & $\binom{n-2}{k-1}-1$ & Theorem~\ref{thm:exp-mds-lb} \\
\hline
\end{tabular}
\caption{Table showing \change{known} lower bounds for the field size of $\change{[}n,k\change{]}$-$\MDS(\ell)$ codes.}\label{tbl:lo}
\end{table}
\renewcommand{\arraystretch}{1}

\change{Table~\ref{tbl:lo} presents the known lower bounds.} \change{Observing} that any $\change{[}n,k\change{]}$-$\MDS(\ell)$ code is also $\MDS(\ell'),\ell'\le \ell$\change{,} we see that our lower bound is an exponential improvement over the earlier best known lower bounds (say when $k=n/2$ and $\ell$ is constant).

\subsection{Open Questions}

We conclude the introduction with the following intriguing open
questions. First, we would like to close the gap between the
existential upper and lower bounds for $\MDS(\ell)$ codes.

\begin{question}
Can we match the existential upper and lower bounds for
$\change{[}n,k\change{]}$-$\MDS(\ell)$ for $\ell \ge 3$? In particular, can we improve the current lower
bound from $\Omega_{k}(n^k)$ to $\Omega_{\ell,k}(n^{\Omega(\ell k)})$?
\end{question}

Second, we would like to close the gap between the existential upper
bounds and explicit constructions.

\begin{question}
Can we construct explicit $\change{[}n,k\change{]}$-$\MDS(\ell)$ codes with field size
$O_{\ell,k}(n^{(\ell-1)k})$?
\end{question}

Finally, we would like to get a truly optimal construction in the
$\change{[}n,3\change{]}$-$\MDS(3)$ case.

\begin{question}
Do there exist $\change{[}n,3\change{]}$-$\MDS(3)$ codes over fields of size $O(n^2)$?
\end{question}

Because of the rich connections higher order MDS codes, progress on any of these questions, particularly the constructions of explicit codes, will lead to new insights and applications for distributed storage and list decoding.

\subsection{Subsequent Work}\label{sec:sub}

Since this paper was first posted, there have been a number of follow-up works. An improved upper bound on the field size of $\change{[}n,k\change{]}$-$\MDS(\ell)$ codes was proved by \change{Athi, Chigullapally, Krishnan, and Lalitha, V}~\cite{athi2023structure} for some choices of $n,k,\ell$. Theorems~\ref{thm:exp-mds-lb} and~\ref{thm:ldlb} were generalized by \change{Alrabiah, Guruswami, and Li}~\cite{alrabiah2023ag} to show that any code, linear or non-linear, which is $\eps$-close to achieving list-decoding capacity requires a field size of at least $\exp(\Omega(1/\eps))$. They prove this result by extending the techniques of this paper. In another line of work, Theorem~\ref{thm:ldlb} has been circumvented entirely \cite{guo2023randomly,alrabiah2023randomly,bdg2023b} by considering suitable \emph{relaxations} of higher order MDS codes in order to construct linear codes which achieve (average-radius) list-decoding capacity over polynomial-sized or constant-sized fields.

\change{\paragraph{Acknowledgments.}
We thank anonymous reviewers for many helpful comments.}

\section{Preliminaries}

In this section, we state a couple of definitions and known theorems about higher order
MDS codes which will be useful for us.

An $\change{[}n,k\change{]}$-Reed-Solomon code with \emph{generators} $\alpha_1, \hdots,
\alpha_n \in \F$ is the code with the following \change{generator matrix}:
\[
  \begin{pmatrix}
    1 & 1 & \cdots & 1\\
    \alpha_1 & \alpha_2 & \cdots & \alpha_n\\
    \alpha_1^2 & \alpha_2^2 & \cdots & \alpha_n^2\\
    \vdots & \vdots & \ddots & \vdots\\
    \alpha_1^{k-1} & \alpha_2^{k-1} & \cdots & \alpha_n^{k-1}
  \end{pmatrix}.
\]

We now discuss properties of general higher order MDS codes.

\begin{theorem}[\cite{bgm2022}]\label{thm:mds5}
  Let $V$ be a $k \times n$ matrix. Then, $V$ is $\change{[}n,k\change{]}$-$\MDS(\ell)$
  if and only if for all $A_1, \hdots, A_\ell \subseteq [n]$ with
  $|A_i| \le k$ and $|A_1| + \cdots + |A_\ell| = (\ell-1)k$, we have
  that $V_{A_1} \cap \cdots \cap V_{A_\ell} = 0$ whenever for all partitions
  $P_1 \cup \cdots \cup P_s = [\ell]$ we have that
  \begin{align}
    \sum_{i=1}^s \left|\bigcap_{j \in P_i} A_j \right| \le (s-1)k.\label{part:ineq}
  \end{align}
\end{theorem}

We \emph{prove} the following simple corollary for $\MDS(3)$\change{, which shows we can ignore that case that $|A_i| = k$ for some $i \in [3]$.}

\begin{corollary}\label{cor:red}
  Let $V$ be a $k \times n$ \change{$\MDS$} matrix. Then, $V$ is $\change{[}n,k\change{]}$-$\MDS(3)$
  if and only if for all $A_1, A_2, A_3 \subseteq [n]$ with
  $|A_i| \le k-1$ and $|A_1| + |A_2| + |A_3| = 2k$, we have that
  $V_{A_1} \cap V_{A_2} \cap V_{A_3} = 0$ whenever it generically
  should; that is, the following conditions hold:
  \begin{itemize}
    \item $|A_1 \cap A_2 \cap A_3| = 0$.
    \item $|A_{\pi(1)} \cap A_{\pi(2)}| + |A_{\pi(3)}| \le k$ for all
      permutations $\pi : [3] \to [3]$.
  \end{itemize}
\end{corollary}

\begin{proof}
  \change{First,} the combinatorial conditions on the sets $A_1, A_2, A_3$ \change{are necessary} by considering suitable partitions of $[3]$ \change{and invoking Theorem~\ref{thm:mds5}}.

 \change{Second, we show that the conditions are sufficient.} By
  Theorem~\ref{thm:mds5}, it suffices to check \change{when one of the $A_i$'s has size
  exactly $k$} \change{that}
  $V_{A_1} \cap V_{A_2} \cap V_{A_3}\change{ = 0}$  \change{if and only if (\ref{part:ineq}) holds}. Assume without loss of generality that $|A_3| =
  k$. Then $V_{A_3} = \F^k$, so $\change{V_{A_1} \cap V_{A_2} \cap V_{A_3} = } V_{A_1} \cap V_{A_2}$. Since $V$ is $\MDS$, $V$ is also $\MDS(2)$
  (see \cite{bgm2021mds}), so $\change{V_{A_1} \cap V_{A_2} =} 0$ if and only if
  it generically should.
\end{proof}

We shall also use the following matrix identity for checking
$\MDS(\ell)$ conditions.

\begin{lemma}[\cite{tian2019formulas,bgm2021mds}]\label{lem:mds5}
  Let $V$ be a $k \times n$ matrix. Consider
  $A_1, \hdots, A_\ell \subseteq [n]$ with $|A_i| \le k$ and
  $|A_1| + \cdots + |A_\ell| = (\ell-1)k$, we have that
  $V_{A_1} \cap \cdots \cap V_{A_\ell} = 0$ if and only if
  \[\det \begin{pmatrix}
    I_k & V\change{|}_{A_1} & & & \\
    I_k & & V\change{|}_{A_2} & &\\
    \vdots & & & \ddots &\\
    I_k & & & &  V\change{|}_{A_{\change{\ell}}}
  \end{pmatrix} \neq 0,\]
  where $V\change{|}_{A_i}$ denotes the submatrix of $V$ with columns indexed by $A_i$.
\end{lemma}

The following is useful for constructing Reed-Solomon codes over $\change{[}n,3\change{]}$-$\MDS(3)$.

\begin{lemma}[\cite{bgm2021mds,roth2021higher}]\label{lemma:mds3}
Let $V$ be an $\change{[}n,3\change{]}$-Reed-Solomon code with evaluation points $\beta_1,
\hdots, \beta_n \in \F$. Then, $V$ is $\MDS(3)$ if and only if for all injective
maps $\alpha : [6] \to [n]$ we have that
\[
  \det \begin{pmatrix} 1 & \beta_{\alpha(1)} + \beta_{\alpha(2)} &
    \beta_{\alpha(1)}\beta_{\alpha(2)}\\
1 & \beta_{\alpha(3)} + \beta_{\alpha(4)} &
    \beta_{\alpha(3)}\beta_{\alpha(4)}\\1 & \beta_{\alpha(5)} + \beta_{\alpha(6)} &
    \beta_{\alpha(5)}\beta_{\alpha(6)}\\
  \end{pmatrix} \neq 0.
\]
\end{lemma}

We shall also use a variant of the GM-MDS theorem.

\begin{theorem}[GM-MDS
  theorem\change{, }\cite{dau2014gmmds,lovett2018gmmds,yildiz2019gmmds}, see \cite{bgm2022}]\label{thm:gm-mds}
Let $A_1, \hdots, A_\ell$ of total size $(\ell-1)k$ such that
  $W_{A_1} \cap \cdots \cap W_{A_\ell} = 0$ for a generic
  \change{$k \times n$ matrix}. Let $\F$ be a field of size at least $n+k-1$. Then,
  there exists $\gamma_1, \gamma_2, \hdots, \gamma_n \in \F$ such that
  the $\change{[}n,k\change{]}$ Reed-Solomon code with generator matrix $U$ with generators $\gamma_1, \hdots,
  \gamma_n$ has that $U_{A_1} \cap \cdots \cap U_{A_\ell} = 0$.
\end{theorem}

\begin{remark}
Note the order of the quantifiers. The code $U$ is only guaranteed to
meet the $\MDS(\ell)$ criteria for one tuple of sets $(A_1, \hdots,
A_\ell)$.
\end{remark}

\section{\texorpdfstring{$\change{[}n,k\change{]}$-$\MDS(3)$}{[n,k]-MDS(3)} codes require field
  size  \texorpdfstring{$\binom{n-2}{k-1}-1$}{Choose(n-2,k-1)-1}}

In this section, we show the exponential (in dimension) lower bound for the field size of $\MDS(3)$ codes.

\begin{reptheorem}{thm:exp-mds-lb}
Let $V \in \F^{k \times n}$ be an $\MDS(3)$-code. Then, $|\F| \ge \binom{n-2}{k-1}-1$.
\end{reptheorem}

We prove this by in fact showing a slightly stronger lower bound.

\begin{lemma}\label{lem:2k-1k-1}
Let $V \in \F^{k \times n}$ be an $\MDS$ code such that $V_{A_1} \cap V_{A_2} \cap V_{A_3}=0$ for all $A_1,
A_2, A_3 \subseteq [n]$ with distinct $A_2,A_3$, $|A_1| = 2, |A_2| = |A_3| = k-1$ and $A_1
\cap (A_2 \cup A_3) = \emptyset$. Then, $|\F| \ge \binom{n-2}{k-1}-1$.
\end{lemma}

\begin{proof}[Proof of Theorem~\ref{thm:exp-mds-lb}]
From Corollary~\ref{cor:red}, consider any $A_1, A_2, A_3$ with $|A_1| = 2, |A_2| = |A_3| = k-1$ and $A_1
\cap (A_2 \cup A_3) = \emptyset$ and $|A_2 \cap A_3| + |A_1| \le k$. Then, we have that $A_2 \neq A_3$ and $V_{A_1} \cap V_{A_2} \cap
V_{A_3} = 0$ when $V$ is $\MDS(3)$. The field size lower bound then follows from Lemma~\ref{lem:2k-1k-1}. 
\end{proof}

\begin{proof}[Proof of Lemma~\ref{lem:2k-1k-1}]
Let $e_1,e_2,\dots,e_k$ be the coordinate vectors in $\F^k.$
Via row operations, we may assume that $V_1 = e_1$ and $V_2 = e_2$. Now consider the
$\MDS(3)$ test on $A_1 = \{1,2\}$ and two distinct subsets $A_2, A_3$
of size $k-1$ which are disjoint from $A_1$ (but not necessarily from
each other). By Corollary~\ref{cor:red}, we have that $V_{A_1} \cap
V_{A_2} \cap V_{A_3}=0$.

Note that for any subset $A \subseteq [n]$ of size $k-1$,
$V_A^{\perp}$ is a $1$-dimensional space with coordinates $(w^A_1,
\hdots, w^A_k)$, where $w^A_i = (-1)^i\det(V|_{([k] \setminus
  i) \times A})$.\footnote{The submatrix with rows indexed by $[k]\setminus i$ and columns indexed by $A$.} Also $V_{A_1}^\perp=\Span\{e_3,e_4,\dots,e_k\}$. Finally $V_{A_1} \cap V_{A_2} \cap V_{A_3}=0$ iff $V_{A_1}^\perp+V_{A_2}^\perp +V_{A_3}^\perp=\F^k.$ Therefore, we have that the $\MDS(3)$ condition on $A_1,A_2,A_3$
is equivalent to 
\[
  \det \begin{pmatrix}
    0 & 0 & \cdots & 0 & 0 & w^{A_2}_1 & w^{A_3}_1\\
    0 & 0 & \cdots & 0 & 0 & w^{A_2}_2 & w^{A_3}_2\\
    1 & 0 & \cdots & 0 & 0 & w^{A_2}_3 & w^{A_3}_3\\
    0 & 1 & \cdots & 0 & 0 & w^{A_2}_4 & w^{A_3}_4\\
    \vdots & \vdots & \ddots & \vdots & \vdots & \vdots & \vdots\\
    0 & 0 & \cdots & 1 & 0 & w^{A_2}_{k-1} & w^{A_3}_{k-1}\\
    0 & 0 & \cdots & 0 & 1 & w^{A_2}_{k} & w^{A_3}_{k}.
  \end{pmatrix} = \det \begin{pmatrix}w^{A_2}_1 & w^{A_3}_1\\
    w^{A_2}_2 & w^{A_3}_2\end{pmatrix} \neq 0.
\]

In other words, $(w_1^{A_2} : w_2^{A_2}) \neq (w_1^{A_3} :
w_2^{A_3})$ in $\PF^1$ whenever $A_2,A_3 \subseteq \{3,\hdots,
n\}$ are distinct subsets of size $k-2$. Since $\PF^1$ has $q+1$ distinct points, we have that $q \ge \binom{n-2}{k-1}-1$. 
\end{proof}

This result greatly improves on the previous lower bounds for MR tensor
codes.

\begin{corollary}
  Let $U \otimes V$ be an $(m,n,a,b)$ MR tensor code with $a \ge 1$
  and $m-a \ge 2$. Then the field size must be at least
  $\binom{n-2}{b-1}-1$.
\end{corollary}

\begin{proof}
Let $U'$ be a $(m-a+1, m-a)$ code formed by puncturing $U$. Note that
$U' \otimes V$ must be an $(m-a+1,n,1,b)$ MR tensor code. Thus, $V$
must be an $\MDS(m-a+1)$ code~\cite{bgm2021mds}. Since $m-a+1 \ge 3$, we have that $V$
is an $\MDS(3)$ code. Thus, the field size lower bound of
$\binom{n-2}{(n-b)-1} -1 = \binom{n-2}{b-1}-1$ applies.
\end{proof}

\subsection{Application to list decoding}
As mentioned previously the lower bound of Theorem~\ref{thm:exp-mds-lb} is only about average-radius list decoding. Furthermore, the ``hard'' case identified is
a rather extreme example, where one of the codewords has Hamming
distance $n-k-2$ from the received codeword, which is only three
less than the minimum distance.

In this section we prove a lower bound for the worst-case list decoding setting by proving Theorem~\ref{thm:ldlb}.

\begin{reptheorem}{thm:ldlb}
Let $n \ge k \ge 0$ be such that $n-k$ is divisible by $3$. Let $C$ be
an $\change{[}n,k\change{]}$-MDS code which is $(2, \frac{2(n-k)}{3n})$ worst-case list decodable. Then, $C$ requires field
size $\binom{n-2(n-k)/3}{k-1} - 1.$
\end{reptheorem}

To start, we show that the $\MDS(3)$ lower bound extends to the case
the sets $A_1, A_2, A_3$ are all the same size.

\begin{lemma}[Extension to the $(2k/3, 2k/3, 2k/3)$ split.]\label{lem:splitlb}
  Let $k$ be divisible by $3$. Let $V$ be an $\change{[}n,k\change{]}$-$\MDS$ code such
  that for all $A_1, A_2, A_3 \subseteq [n]$ of size $2k/3$ we have
  that $V_{A_1} \cap V_{A_2} \cap V_{A_3} = 0$ whenever this holds
  generically. Then, the field size of the code is at least
  $\binom{n-2k/3}{k/3+1} - 1.$ %
\end{lemma}

\begin{proof}
  Let $V$ be an $\MDS(3)$ code such that
  $V_{A_1} \cap V_{A_2} \cap V_{A_3} = 0$ whenever
  $|A_1| = |A_2| = |A_3| = 2k/3$ and generically the intersection
  should be $0$. By Corollary~\ref{cor:red}, this is equivalent to
  $|A_1 \cap A_2 \cap A_3| = 0$ and $|A_i \cap A_j| \le k/3$ for all
  $i \neq j$. %

Pick $2k/3-2$ columns $I \subseteq [n]$ of $V$.  Let $k' = k - |I| =  k/3+2$ Pick an arbitrary projection $\Pi :
\F^k \to \F^{k'}$ of rank $k'$ such that $\ker(\Pi) = V_I$. Let $V'$ be the code with columns $\Pi(v_{i})$ for
$i \in [n] \setminus I$.

\begin{claim}
$V'$ is MDS.
\end{claim}
\begin{proof}
Pick any set of indices $J$ of size $k'$ disjoint from $I$, it suffices to prove that
$V'_{I} = \F^{k'}$. Note that
\begin{align*}
  V'_I &= \Span \{\Pi(v_i) : i \in J\}\\
       &= \Span \{\Pi(v_i) : i \in I \cup J\}\\
       &= \Pi(\Span \{v_i : i \in I \cup J\})\\
       &= \Pi(\F^k) && \text{ ($V$ MDS)}\\
       &= \F^{k'} && \text{ ($\Pi$ max rank)},
\end{align*}
as desired.
\end{proof}

To complete the field size lower bound, consider distinct
$A'_1, A'_2, A'_3 \subseteq [n]\setminus I$ such that $|A'_1| = 2$,
$|A'_2|=|A'_3| = k'-1 = k/3+1$, and
$A'_1 \cap (A'_2 \cup A'_3) = 0$. By Lemma~\ref{lem:2k-1k-1}, it suffices to show that
$V'_{A'_1} \cap V'_{A'_2} \cap V'_{A'_3} = 0$. Assume for sake of
contradiction that there exists nonzero
$v_0 \in V'_{A'_1} \cap V'_{A'_2} \cap V'_{A'_3}$. Pick an arbitrary
partition $I_2 \cup I_3 = I$ such that $|I_2| = |I_3| = k/3-1$. Now
let
\begin{align*}
  A_1 &= I \cup A'_1\\
  A_2 &= I_2 \cup A'_2\\
  A_3 &= I_3 \cup A'_3.
\end{align*}

Note that $|A_1| = |A_2| = |A_3| = 2k/3$ and $|A_1 \cap A_2 \cap A_3| =
0$. Further, we can check that
\begin{align*}
  |A_1 \cap A_2| + |A_3| &= |I_2| + |A_3| = k/3-1 + 2k/3 \le k,\\
  |A_1 \cap A_3| + |A_2| &= |I_3| + |A_2| = k/3-1 + 2k/3 \le k,\\
  |A_2 \cap A_3| + |A_1| &= |A'_2 \cap A'_3| + |A_1| \le k/3+2k/3\le k.
\end{align*}

Thus, by Corollary~\ref{cor:red}, we know that $V_{A_1} \cap V_{A_2}
\cap V_{A_3} = 0$. 

Let $W = \Pi^{-1}(\Span(v_0))$ (recall that $v_0 \in V'_{A'_1} \cap
V'_{A'_2} \cap V'_{A'_3}$). %
Note that $\dim(W) = k-k'+1 = 2k/3-1$. Since $\Span(v_0) \in V'_{A'_1}$, we
have that 
\[
  W = \Pi^{-1}(\Span(v_0)) \subseteq \Pi^{-1}(V'_{A'_1}) = V_I +
  V_{A'_1} = V_{A_1}.
\]
Also observe that since $v_0 \in V'_{A'_2} = \Pi(V_{A_2})$, there
exists $w_0 \in V_{A_2}$ such that $\Pi(w_0) = v_0$. Thus,
$w_0 \in W$, too, so $v_0 \in \Pi(W \cap V_{A_2})$. Further, let
$\Pi_2$ be the map $\Pi$ but restricted to the domain
$W \cap V_{A_2}$. Observe that by the rank-nullity theorem
\begin{align*}
  \dim (W \cap V_{A_2})
  &= \dim \ker \Pi_2 + \dim (\Pi(W \cap V_{A_2}))\\
  &\ge \dim V_{I_2} + \dim \Span (v_0)\\
  &= |I_2| + 1.
\end{align*}
Likewise, $\dim(W \cap V_{A_3}) \ge |I_3| + 1$.
Thus,
\begin{align*} \dim(V_{A_1} \cap V_{A_2} \cap V_{A_3})
  &\ge \dim(W \cap V_{A_2} \cap V_{A_3})\\
  &= \dim((W \cap V_{A_2}) \cap (W \cap V_{A_3}))\\
  &\ge \dim(W \cap V_{A_2}) + \dim(W \cap V_{A_3}) - \dim(W)\\
  &\ge |I_2|+1 + |I_3|+1 - |W|\\
  &= 1,
\end{align*}
a contradiction. %
Therefore, by Lemma~\ref{lem:2k-1k-1} we get a lower bound of
\[
  q \ge \binom{n-2k/3}{k/3+1} - 1.\qedhere
\]

\end{proof}

We are now ready to prove Theorem~\ref{thm:ldlb}.

\begin{proof}[Proof of Theorem~\ref{thm:ldlb}]
We adapt the proof of Proposition~4.1 in \cite{bgm2022}. Let
$H$ be the parity check matrix of $C$ (i.e., the generator matrix of
$C^{\perp}$). If $H$ has the property that for all $A_1, A_2, A_3
\subseteq [n]$ of size $\frac{2}{3}(n-k)$ we have that $H_{A_1}
\cap H_{A_2} \cap H_{A_3} = 0$ when it generically should, then the
desired field size lower bound holds by Lemma~\ref{lem:splitlb}. 

Otherwise, there exists $A_1, A_2, A_3$ of size $\frac{2}{3}(n-k)$
such that $H_{A_1} \cap H_{A_2} \cap H_{A_3} \neq 0$ even though the
generic intersection is $0$. %
Let $v_0$ be a nonzero vector in this common
intersection. By definition, for each $i \in [3]$ there exists a
nonzero $e_i \in \F^n$ such that $v_0 = He_i$ and
$\supp(e_i) \subseteq A_i$.

If $e_1, e_2, e_3$ are all distinct, then we have violated that $C$ is
$(2, \frac{2(n-k)}{3n})$ worst-case list decodable (see equation 17 of
\cite{bgm2022}). If $e_1 = e_2 = e_3$, then
$A_1 \cap A_2 \cap A_3 \neq 0$ which contradicts
Corollary~\ref{cor:red}. Otherwise, if say WLOG $e_1 = e_2 \neq e_3$,
then $\supp(e_1),\supp(e_2)\subset A_1\cap A_2$. So, by
Corollary~\ref{cor:red},
\[
  \lvert\supp(e_1 - e_3)\rvert \le |A_1 \cap A_2| + |A_3| \le k.
\]
However $H(e_1 -e_3) = 0$, so $C$ cannot be MDS. \qedhere
\end{proof}
\section{Doubly-exponential construction of higher-order MDS codes}

\change{Shangguan and Tamo~\cite{shangguan2020combinatorial} (c.f., Theorem 1.7)} give a doubly-exponential explicit
construction of higher order MDS codes. In particular, they construct
$\MDS(3)$ codes which have field size $2^{k^n}$ and $\MDS(4)$ codes which have field size $2^{(3k)^n}$, although their method
can adapted to $\MDS(\ell)$ codes (see their remark after Theorem
1.9). \change{Roth~\cite{roth2021higher}} improves this construction to
$n^{k^{O(k)}}$ for $\MDS(3)$ codes. 

\begin{remark}
We also note that a straightforward construction can be obtained by
using ideas from \change{the works of Brakensiek-Gopi-Makam~\cite{bgm2022} and
Shangguan-Tamo \cite{shangguan2020combinatorial}}. For a $k\times n$ matrix $V$, using
the equivalence of $V_{A_1}\cap\hdots\cap V_{A_\ell}=0$ with
$V^\perp_{A_1}+\hdots+V_{A_\ell}^\perp=\F_q^k$, we can write down a
determinantal matrix identity which for Reed-Solomon
codes only has individual degree at most $k-1$ in the evaluation points~\cite{bgm2022}. This will give us a general $\change{[}n,k\change{]}$-$\MDS(\ell)$ construction over fields of size $2^{k^n}$ by working over the field $\F_2(x_1,\hdots,x_n)$ where $x_i$ has degree $k$ over $\F_2(x_1,\hdots,x_{i-1})$, similar to what was done \change{by Shangguan-Tamo} \cite{shangguan2020combinatorial}.
\end{remark}
In the remainder of this section, we construct
$\MDS(\ell)$ codes for all $\ell \ge 4$ with field sizes comparable to
Roth's.

\begin{reptheorem}{thm-con-gen}
There is an explicit $\change{[}n,k\change{]}$-$\MDS(\ell)$-code over field size
$n^{(\ell k)^{O(\ell k)}}$.
\end{reptheorem}

The proof incorporates ideas from both
Shangguan-Tamo~\cite{shangguan2020combinatorial} and
Roth~\cite{roth2021higher}.

\begin{proof}
Pick $F_0$ to be a finite field of size at least $n+k-1$. For $i \in
[\ell k]$, let $F_i$ be a degree $\ell k^2$ extension of $F_{i-1}$ via
generator $\alpha_i$. Let $\F = F_{\ell k}$. Note that $|\F| = n^{(\ell
  k)^{O(\ell k)}}$. Pick distinct $\beta_1, \hdots, \beta_n \in F_0$. For
all $i \in [n]$, define the multivariate polynomial
\[
  p_i(x_1, \hdots, x_{\ell k}) = \sum_{j=1}^{\ell k} \beta_i^{j-1} x_j.
\]

Let $V$ be the $\change{[}n,k\change{]}$-Reed-Solomon code with generators
$p_i(\alpha_1, \hdots, \alpha_{\ell k})$ for $i \in [n]$. We claim
that $V$ is $\MDS(\ell)$. Let $\widetilde{V}$ be an
$\change{[}n,k\change{]}$-Reed-Solomon code with generators $p_i(x_1, \hdots, x_{\ell
  k})$, where the base field is $\widetilde{\F} := F_0(x_1, \hdots, x_{\ell k})$ (i.e.,
$F_0$ extended by $\ell k$ free generators). 

The desired result follows by the following two claims.

\begin{claim}
  $V$ is $\MDS(\ell)$ if and only if $\widetilde{V}$ is $\MDS(\ell)$.
\end{claim}

\begin{proof}
  To verify, consider $A_1, \hdots, A_\ell$ of total size $(\ell-1)k$
  such that $W_{A_1} \cap \cdots \cap W_{A_\ell} = 0$ for a generic
  \change{$k \times n$ matrix}. It suffices to check by Lemma~\ref{lem:mds5} that

\begin{align}
  \det \begin{pmatrix}
    I_k & V\change{|}_{A_1} & & & \\
    I_k & & V\change{|}_{A_2} & &\\
    \vdots & & & \ddots &\\
    I_k & & & &  V\change{|}_{A_{\change{\ell}}}
  \end{pmatrix} \neq 0 &\iff \det \begin{pmatrix}
    I_k & \widetilde{V}\change{|}_{A_1} & & & \\
    I_k & & \widetilde{V}\change{|}_{A_2} & &\\
    \vdots & & & \ddots &\\
    I_k & & & &  \widetilde{V}\change{|}_{A_{\change{\ell}}}
  \end{pmatrix} \neq 0 \label{eq:oes}
\end{align}

Note that each term in the determinant is the product of at most
$\ell k$ expressions of the form $p_i^{j}$ for $j \le k-1$. Since each
$\alpha_i$ is from a degree $\ell k^2$ extension, we have that the LHS
determinant is nonzero in $F_0[\alpha_1, \hdots, \alpha_{\ell k}]$ iff RHS determinant is nonzero in
$F_0(x_1, \hdots, x_{\ell k})$.
\end{proof}

\begin{claim}
  $\widetilde{V}$ is $\MDS(\ell)$.
\end{claim}
\begin{proof}
 Again, consider
  $A_1, \hdots, A_\ell$ of total size $(\ell-1)k$ such that
  $W_{A_1} \cap \cdots \cap W_{A_\ell} = 0$ for a generic
  \change{$k \times n$ matrix}. Since $|F_0| \ge n+k-1$, Theorem~\ref{thm:gm-mds},
  there exists $\gamma_1,\gamma_2,\dots,\gamma_k\in F_0$ such that the $\change{[}n,k\change{]}$
  Reed-Solomon code $U$ generated by $\gamma_1, \hdots, \gamma_n$ has
  that $U_{A_1} \cap \cdots \cap U_{A_\ell}=0$.
   Thus, it suffices to
  prove that there exists an assignment
  $\pi : \{x_1, \hdots, x_{\ell k}\} \to F_0$ such that
  \[
    \gamma_i = p_i(\pi(x_1), \hdots, \pi(x_{\ell k})).
  \]
  This shows that the RHS of (\ref{eq:oes}) holds. Finding such a $\pi$
  is equivalent to solving the following linear system:
  \[
    \begin{pmatrix}
      1 & \beta_1 & \cdots & \beta_1^{\ell k - 1}\\
      1 & \beta_2 & \cdots & \beta_2^{\ell k - 1}\\
      \vdots & \vdots & \ddots & \vdots\\
      1 & \beta_{\ell k} & \cdots & \beta_{\ell k}^{\ell k - 1}\end{pmatrix}
    \begin{pmatrix}
      \pi(x_1)\\
      \pi(x_2)\\
      \vdots\\
      \pi(x_{\ell k})\\
    \end{pmatrix}
    =
    \begin{pmatrix}
      \gamma_1\\
      \gamma_2\\
      \vdots\\
      \gamma_{\ell k}
    \end{pmatrix}
  \]
  Since the square matrix on the LHS is a Vandermonde matrix, it is
  invertible, so such a $\pi$ does indeed exist.
\end{proof}

Thus, $V$ is indeed $\MDS(\ell)$.
\end{proof}
\begin{remark}
  Using techniques, we can also get an explicit construction of
  $\change{[}n,n-b\change{]}$-$\MDS(\ell)$ codes over fields of size
  $n^{{\ell b}^{O(\ell b)}}$ by adapting the non-constructive upper
  bound \change{by Brakensiek, Gopi, and Makam}~\cite{bgm2021mds} (see their Appendix A). In particular, the
  $\MDS(\ell)$ conditions are equivalent to the maximal recoverability
  of a suitable tensor code. By looking at the parity check matrix of
  the tensor code, these recoverability conditions can be expressed as
  ensuring nonzero determinants of matrices of size at most
  $(\ell-1)b$. Likewise, one can show the degree of each variables in
  the determinant is $\poly(\ell, b)$, yielding the stated bound.
\end{remark}
\section{New explicit constructions for \texorpdfstring{$\change{[}n,3\change{]}$-$\MDS(3)$}{[n,3]-MDS(3)}}\label{sec:n3mds3}

First, we present a hand-verifiable $O(n^4)$ construction for
$\change{[}n,3\change{]}$-$\MDS(3)$. After that we present an $O(n^3)$ construction, where a portion
of the proof involves computation of a Groebner basis (for more
details, refer to \cite{cox2013ideals}). We compute the Groebner bases
in Julia~\cite{bezanson2017julia} using the OSCAR library~\cite{OSCAR,OSCAR-book}. We provide the code we used to
perform these computations.\footnote{\url{https://github.com/jbrakensiek/MDS3-Groebner}}

\subsection{A simple \texorpdfstring{$O(n^4)$}{O(n 4)} construction}

\begin{theorem}\label{thm:n4}
There exists an explicit $\change{[}n,3\change{]}$-$\MDS(3)$ code with field size $O(n^4)$.
\end{theorem}

\begin{proof}
Assume $q$ is an odd prime power of size at least $n$. Let
$\F_q[\gamma]$ be a degree four extension. Let
$\alpha : [n] \to \F_q$ be some injective map. We claim that
\[
  \beta_i = \alpha(i) + \gamma \alpha(i)^2, i \in [n]
\]
form the generators of a $\change{[}n,3\change{]}-\MDS(3)$ code. In particular, by
Lemma~\ref{lemma:mds3}, we must
check that
\[
  \det \begin{pmatrix}
    1 & \beta_{1} + \beta_{2} & \beta_{1}\beta_{2}\\
    1 & \beta_{3} + \beta_{4} & \beta_{3}\beta_{4}\\
    1 & \beta_{5} + \beta_{6} & \beta_{5}\beta_{6}
  \end{pmatrix} \neq 0
\]
for all injective $\alpha : [6] \to \F_q$. When expanded, LHS is a degree
three polynomial in $\gamma$; that is, the symbolic determinant can be
written as
\[
  P(\gamma) = p_0(\alpha)\gamma^0 + p_1(\alpha)\gamma^1 + p_2(\alpha)\gamma^2 + p_3(\alpha)\gamma^3,
\]
where $p_i$ is some polynomial in $\F_q[x_1, \hdots,
x_6]$ evaluated at $\alpha(1), \hdots, \alpha(6)$. A notable observation that will be useful later is that
\[
  p_1(\alpha) = p_0(\alpha)(\alpha_1 + \cdots + \alpha_6).
\]
Since $\gamma$ is a degree four field extension, we have that
$P(\gamma) = 0$ if and only if $p_i(\alpha) = 0$ for all
$i \in \{0,1,2,3\}$. Let $J \subseteq \F_q[x_1, \hdots, x_6]$ be the
ideal generated by $p_0, p_1, p_2, p_3$. It suffices to prove that
there exists $h \in J$ such that $h(\alpha) \neq 0$. We show this as
follows.

\begin{claim}
Let $S = \{(3,6),(2,4),(2,6),(3,5),(4,5),(4,6),(2,5),(2,3)\}$. Let
$h = \prod_{(i,j) \in S}(x_i - x_j)$. Then, $h \in J$.
\end{claim}

\begin{proof}
It suffices to find $Q_0, Q_1, Q_2, Q_3 \in \F_q[x_1, \hdots, x_6]$
such that $h = Q_0p_0 + Q_1p_1 + Q_2p_2 + Q_3p_3$. Note since $p_0$
divides $p_1$, we may assume $Q_1 = 0$.

Now, let $g(x) = x_2x_3 +x_2x_4 - x_2x_5 -x_2x_6  -x_3x_4 +
x_5x_6$. We have that
\begin{align*}
Q_2 &= x_2 \cdot g\\
Q_3 &= -\frac{1}{2}g
\end{align*}
Finally,
\begin{align*}
  Q_0 &= \frac{1}{2}(-2x_1x_2^3x_3 - 2x_1x_2^3x_4 + 2x_1x_2^3x_5 + 2x_1x_2^3x_6 -
        x_1x_2^2x_3^2 - x_1x_2^2x_4^2 + x_1x_2^2x_5^2\\
      &\hspace{0cm} + x_1x_2^2x_6^2 + 2x_1x_2x_3^2x_4 + 2x_1x_2x_3x_4^2 -
        2x_1x_2x_5^2x_6 - 2x_1x_2x_5x_6^2 - x_1x_3^2x_4^2\\
      &\hspace{0cm}+ x_1x_5^2x_6^2 - 2x_2^3x_3^2 - 2x_2^3x_3x_4 -
        2x_2^3x_4^2 + 2x_2^3x_5^2 + 2x_2^3x_5x_6\\
      &\hspace{0cm}+ 2x_2^3x_6^2 - x_2^2x_3^2x_4 - x_2^2x_3x_4^2 +
        x_2^2x_3x_4x_5 + x_2^2x_3x_4x_6 - x_2^2x_3x_5x_6\\
      &\hspace{0cm}- x_2^2x_4x_5x_6 + x_2^2x_5^2x_6 + x_2^2x_5x_6^2 +
        3x_2x_3^2x_4^2 + x_2x_3^2x_4x_5 + x_2x_3^2x_4x_6\\
      &\hspace{0cm}- x_2x_3^2x_5x_6 + x_2x_3x_4^2x_5 + x_2x_3x_4^2x_6
        + x_2x_3x_4x_5^2 + x_2x_3x_4x_6^2\\
      &\hspace{0cm}- x_2x_3x_5^2x_6 - x_2x_3x_5x_6^2 - x_2x_4^2x_5x_6
        - x_2x_4x_5^2x_6 - x_2x_4x_5x_6^2\\
      &\hspace{0cm}- 3x_2x_5^2x_6^2 - x_3^2x_4^2x_5 - x_3^2x_4^2x_6 +
        x_3^2x_4x_5x_6 + x_3x_4^2x_5x_6\\
      &\hspace{0cm}- x_3x_4x_5^2x_6 - x_3x_4x_5x_6^2 + x_3x_5^2x_6^2 + x_4x_5^2x_6^2)
\end{align*}

\change{We verify the claimed identity in $h$ with Julia code. See the script \textsf{claim-5-2.jl} in the GitHub repository.}

For a simpler perspective, we can make the following substitutions:
\begin{align*}
s_2 &= \alpha_3 + \alpha_4, & p_2 &= \alpha_3\alpha_4,\\
s_3 &= \alpha_5 + \alpha_6, & p_3 &= \alpha_5\alpha_6.
\end{align*}

Then, we have that
\begin{align*}
g &= x_2(s_2-s_3) - p_2+p_3\\
2Q_0 &= -2x_1x_2^3(s_2 - s_3) - x_1x_2^2(s_2^2 - 2p_2 - s_3^2 + 2p_3)\\
      & + 2x_1x_2(s_2p_2 - s_3p_3) - x_1(p_2-p_3)(p_2+p_3)\\
  &- 2x_2^3(s_2^2 - p_2 - s_3^2 + p_3) - x_2^2(s_2-s_3)(p_2+p_3)\\\
  &+x_2(3p_2^2 - 3p_3^2 + (s_2+s_3)(p_2s_3-s_2p_3)) - (p_2s_3 - s_2p_3)(p_2+p_3).
\end{align*}

\end{proof}

Note that $h(\alpha)$ is nonzero, as $\alpha$ is injective. Thus, we
have proved our construction is $\MDS(3)$.

\end{proof}

\begin{remark}
  One can adapt this construction for characteristic $2$. Let
  $F = \F_2[x]/(p(x))$ be a suitable extension field, with $p(x)$ an
  irreducible of degree $\lceil \log_2 n \rceil + 1$. For all
  $i \in [n]$, pick distinct $\alpha_i \in F$ such that each
  $\alpha_i$ has $1$ as its $x^0$ coefficient. Pick $\gamma$ from a
  degree four extension of $F$, and let
  $\beta_i = \alpha_i + \gamma \alpha_i^3$. We can define
  $p_0, p_1, p_2, p_3$ as in the proof of Theorem~\ref{thm:n4}, and
  show that the ideal they generate includes
  \[\prod_{i < j \in [6]}(x_i + x_j)\prod_{i < j < k \in [6]}(x_i + x_j +
    x_k) \neq 0,\] proving the code is $\MDS(3)$.
\end{remark}

\subsection{An \texorpdfstring{$O(n^3)$}{O(n 3)} construction.}

With another trick, we can shave the field size by another factor of
$n$. 

\begin{reptheorem}{thm-con-three}
There exists an explicit $\change{[}n,3\change{]}$-$\MDS(3)$ code with field size $O(n^3)$.
\end{reptheorem}

\begin{proof}

Assume $q$ is a power of $7$ and let $\F_q[\gamma]$ be the degree
three extension such that $\gamma^3 = 2$. (Observe that $x^3-2$ is
irreducible over $\F_7$ as it has no roots over $\F_7$.) %
Let $S \subseteq \F_q$ of size $n$ such that no six sum to $0$. It is
clear that we can have $|S| \ge q/7$, so we can do this as long as
$q\ge 7n$.

We claim that
\[
  \beta_\alpha = \alpha + \gamma \alpha^2, \alpha \in S.
\]
form the generators of a $\change{[}n,3\change{]}-\MDS(3)$ code. To check this, consider
an injective map $\alpha : [6] \to S$. Let $\beta_i$ be shorthand for
$\beta_{\alpha(i)}$. By Lemma~\ref{lemma:mds3}, we must check that
\[
  \det \begin{pmatrix}
    1 & \beta_{1} + \beta_{2} & \beta_{1}\beta_{2}\\
    1 & \beta_{3} + \beta_{4} & \beta_{3}\beta_{4}\\
    1 & \beta_{5} + \beta_{6} & \beta_{5}\beta_{6}
  \end{pmatrix} \neq 0
\]
When we expand, without using the identity that $\gamma^3 = 2$, the
LHS is a degree three polynomial in $\gamma$; that is, the symbolic
determinant can be written as
\[
  P(\gamma) = p_0(\alpha)\gamma^0 + p_1(\alpha)\gamma^1 + p_2(\alpha)\gamma^2 + p_3(\alpha)\gamma^3,
\]
where $p_i$ is some polynomial in $\F_q[x_1, \hdots,
x_6]$ evaluated at $\alpha(1), \hdots, \alpha(6)$. Once we apply that
$\gamma^3 = 2$, we get instead that
\[
  P(\gamma) = (p_0(\alpha) + 2p_3(\alpha))\gamma^0 + p_1(\alpha)\gamma^1 + p_2(\alpha)\gamma^2,
\]

Thus, it suffices to check that the ideal
\[
  J := (p_0 + 2 p_3, p_1, p_2)
\]
contains a polynomial which is nonzero when evaluated at $\alpha(1),
\hdots, \alpha(6)$. 

\begin{claim}\label{claim:groebner}
  $(x_1 + \cdots + x_6)\prod_{1 \le i < j \le 6} (x_j - x_i) \in J$.
\end{claim}
\begin{proof}
  We verify this by computing a Groebner basis of $J$ and that the
  remainder upon dividing the LHS by the Groebner basis is $0$. We
  verify this in OSCAR.
\end{proof}

By Claim~\ref{claim:groebner}, it suffices to prove that
\[
  (\alpha(1) + \cdots \alpha(6))\prod_{1 \le i < j \le 6} (\alpha(j) -
  \alpha(i)) \neq 0
\]
for all injective maps $\alpha : [6] \to S$. Note that the first term
in the product is nonzero by the definition of $S$, and the remaining
terms are nonzero since $\alpha$ is injective. Thus, $P(\gamma) \neq
0$, so our code is indeed $\MDS(3)$.
\end{proof}

\section{Constructions of \texorpdfstring{$\change{[}n,4\change{]}$}{[n,4]}-\texorpdfstring{$\MDS(3)$}{MDS(3)}
  and \texorpdfstring{$\change{[}n,5\change{]}$}{[n,5]}-\texorpdfstring{$\MDS(3)$}{MDS(3)} codes}

In this section, we give some structural observations about $\MDS(3)$
codes which shall lead to explicit constructions of $\change{[}n,4\change{]}$-$\MDS(3)$
and $\change{[}n,5\change{]}$-$\MDS(3)$ codes.

\subsection{\texorpdfstring{$\MDS(\ell)$}{MDS(l)} equivalent conditions for Reed-Solomon codes}

First, we give an alternative characterization of the higher-order MDS
conditions when considering a Reed-Solomon code.

Assume we have a $\change{[}n,k\change{]}$-RS code $V$ with evaluation points $\beta_1, \hdots,
\beta_n \in \F$. For $A \subseteq [n]$ define
\[
  \Pi_A(x) = \prod_{i \in A} (x - \beta_i).
\]
Define $\Pi^d_A(x)$ to be the following (row) vector of polynomials:
\[
  \Pi^d_A(x) := (\Pi_A(x), x\Pi_A(x), \cdots, x^{d-1}\Pi_A(x)).
\]

\begin{lemma}\label{lem:prod-mat}
  Assume that $A_1, \hdots, A_\ell \subseteq [n]$ such that
  $|A_i| \le k$ for all $i \in [\ell]$. Let
  $|A_1| + \cdots + |A_\ell| = (\ell-1)k$. Let $\delta_i = k - |A_i|$.
  Assume without loss of generality that
  $A_1 = \{1,2,\hdots, k-\delta_1\}$. We have that
  $V_{A_1} \cap \cdots \cap V_{A_\ell} = 0$ if and only if
  \begin{align}
    \det \begin{pmatrix}
      \Pi^{\delta_2}_{A_2}(\beta_{1}) & \Pi^{\delta_3}_{A_3}(\beta_{1}) &
      \cdots & \Pi^{\delta_\ell}_{A_\ell}(\beta_1)\\
      \Pi^{\delta_2}_{A_2}(\beta_2) & \Pi^{\delta_3}_{A_3}(\beta_2) &
      \cdots & \Pi^{\delta_\ell}_{A_\ell}(\beta_2)\\
      \vdots & \vdots & \ddots & \vdots\\
      \Pi^{\delta_2}_{A_2}(\beta_{k-\delta_1}) & \Pi^{\delta_3}_{A_3}(\beta_{k-\delta_1}) &
      \cdots & \Pi^{\delta_\ell}_{A_\ell}(\beta_{k-\delta_1})
    \end{pmatrix}\label{eq:MDS-key} \neq 0.
  \end{align}
\end{lemma}

\begin{proof}
Starting with the equality $V_{A_1} \cap \cdots \cap V_{A_\ell} = 0,$
take the dual to get
\begin{align}
  V^{\perp}_{A_1} + \cdots + V^{\perp}_{A_\ell} = \F^k. \label{eq:cnp}
\end{align}
Identify each vector $(c_0, \hdots, c_{k-1})^{T} \in \F^k$ with the
polynomial $c_0 + c_1x + \cdots + c_{k-1}x^{k-1}$. Observe that a
degree $\le k-1$ polynomial $p$ is in $V^{\perp}_{A_i}$ if and only if
$p(\beta_i) = 0$ for all $i \in A_i$ which is true if and only if
$\Pi_{A_i} | p$. Thus, the polynomials $(\Pi_{A_i}(x), x \Pi_{A_i}(x),
\hdots, x^{\delta_i-1}\Pi_{A_i}(x)) = \Pi^{\delta_i}_{A_i}(x)$ form a basis of $V_{A_i}^{\perp}$.

Consider the linear transformation $\Lambda : \F^k \to \F^k$ which
sends a polynomial $p$ to the evaluations $(p(\beta_1), \hdots,
p(\beta_k))$. Since the $\beta_i$'s are distinct, this map is
invertible because the Vandermonde matrix is nonsingular. Applying
$\Lambda$ to the basis $\Pi^{\delta_i}_{A_i}(x)$, we have that
$V^{\perp}_{A_i}$ also has the following columns as a basis:
\[
\begin{pmatrix}
  \Pi_{A_i}^{\delta_i}(\beta_1)\\
  \Pi_{A_i}^{\delta_i}(\beta_2)\\
  \vdots\\
  \Pi_{A_i}^{\delta_i}(\beta_k)
\end{pmatrix}.
\]
Thus, (\ref{eq:cnp}) is equivalent to
\begin{align}
\det \begin{pmatrix}
  \Pi^{\delta_1}_{A_1}(\beta_{1}) & \Pi^{\delta_2}_{A_2}(\beta_{1}) &
  \cdots & \Pi^{\delta_\ell}_{A_\ell}(\beta_1)\\
  \Pi^{\delta_1}_{A_1}(\beta_2) & \Pi^{\delta_2}_{A_2}(\beta_2) &
  \cdots & \Pi^{\delta_\ell}_{A_\ell}(\beta_2)\\
  \vdots & \vdots & \ddots & \vdots\\
  \Pi^{\delta_1}_{A_1}(\beta_{k}) & \Pi^{\delta_2}_{A_2}(\beta_{k}) &
  \cdots & \Pi^{\delta_\ell}_{A_\ell}(\beta_{k})
\end{pmatrix}\neq 0.\label{eq:MDS-step}
\end{align}

The upper-left $(k-\delta_i) \times \delta_i$ submatrix in
(\ref{eq:MDS-step}) is all $0$'s since
$A_1=\{\beta_1,\beta_2,\dots,\beta_{k-\delta_1}\}$. The lower-left
$\delta_i \times \delta_i$ submatrix in (\ref{eq:MDS-step}) has full
rank since after removing common factors $\Pi_{A_1}(\beta_i)$ in each
row, we are left with a Vandermonde matrix. Thus, by elementary column
operations, (\ref{eq:MDS-step}) holds if and only if (\ref{eq:MDS-key}) holds.
\end{proof}

We also have the following corollary when $\ell = 3$ that we can WLOG
assume the sets are disjoint.

\begin{corollary}\label{cor:weak}
  Let $A_1, A_2, A_3 \subseteq [n]$ be such that
  $|A_1| + |A_2| + |A_3| = 2k$, $|A_1|, |A_2|, |A_3| \le k$ and
  $|A_1 \cap A_2 \cap A_3| = 0$. Let
  $A'_1 = A_1 \setminus (A_2 \cup A_3), A'_2 = A_2 \setminus (A_1 \cup
  A_3), A'_3 = A_3 \setminus (A_1 \cup A_2)$ with
  $k' = \frac{1}{2}(|A'_1| + |A'_2| + |A'_3|)$. Let $V'$ be the
  $\change{[}n,k'\change{]}$-RS code with the same generators as $V$. Then, $V_{A_1}
  \cap V_{A_2} \cap V_{A_3} = 0$ if and only if $V'_{A_1} \cap
  V'_{A_2} \cap V'_{A_3} = 0$. Furthermore, $k'-|A'_i| \le k-|A_i|$
  for all $i \in [3]$.
\end{corollary}

\begin{remark}
This effectively reduces checking only the $\MDS$ conditions in which
the sets are disjoint, see similar ideas in \cite{bgm2021mds} and
\cite{roth2021higher} (i.e., Theorem 18).
\end{remark}

\begin{proof}
We prove this result by induction on $k$. The base case of $k=0$ is
trivial. For positive $k$, note that if $A_1$, $A_2$ and $A_3$ are
disjoint, the result immediately follows.
Otherwise, assume WLOG that there is some $i
\in A_2 \cap A_3$. Let $\delta_i=k-|A_i|$ and assume WLOG that that $A_1=\{1,2,\dots,k-\delta_1\}$. Note that $i \not\in A_1$ by assumption. We will use Lemma~\ref{lem:prod-mat} to reduce intersection condition to a determinant condition. Now, $x -
\beta_i$ is a factor of both $\Pi_{A_2}$ and $\Pi_{A_3}$. Thus,

\begin{align*}
   \det \begin{pmatrix}
      \Pi^{\delta_2}_{A_2}(\beta_{1}) & \Pi^{\delta_3}_{A_3}(\beta_{1})\\
      \Pi^{\delta_2}_{A_2}(\beta_2) & \Pi^{\delta_3}_{A_3}(\beta_2)\\
      \vdots & \vdots &\\
      \Pi^{\delta_2}_{A_2}(\beta_{k-\delta_1}) & \Pi^{\delta_3}_{A_3}(\beta_{k-\delta_1})
    \end{pmatrix} = \prod_{j=1}^{k-\delta_1}(\beta_j - \beta_i)\det \begin{pmatrix}
      \Pi^{\delta_2}_{A_2\setminus \{i\}}(\beta_{1}) & \Pi^{\delta_3}_{A_3\setminus\{i\}}(\beta_{1})\\
      \Pi^{\delta_2}_{A_2\setminus \{i\}}(\beta_2) & \Pi^{\delta_3}_{A_3\setminus\{i\}}(\beta_2)\\
      \vdots & \vdots &\\
      \Pi^{\delta_2}_{A_2\setminus \{i\}}(\beta_{k-\delta_1}) &
      \Pi^{\delta_3}_{A_3\setminus \{i\}}(\beta_{k-\delta_1})
    \end{pmatrix}.
  \end{align*}
  Since $\beta_i \neq \beta_j$ for all $j \in A_1$. We have that one
  of the two determinants is nonzero if and only if the other is
  nonzero. Let $V'$ be the $\change{[}n,k-1\change{]}$-RS code with the generators
  $\beta_1, \hdots, \beta_n$. The determinant on the RHS being nonzero
  is equivalent to whether $V'_{A_1} \cap V'_{A_2 \setminus \{i\}} \cap
  V'_{A_3 \setminus \{i\}} = 0$. Note that $(k-1)-|A_1| \le k -|A_1|$,
  $(k-1)-|A_2 \setminus \{i\}| \le k -|A_2|$ and $(k-1)-|A_3 \setminus
  \{i\}| \le k - |A_3|.$ Thus, by induction, we may iterate on
  $(A_1, A_2 \setminus \{i\}, A_3 \setminus \{i\})$ until the three
  sets are pairwise disjoint. 
\end{proof}

\subsection{\texorpdfstring{$\change{[}n,4\change{]}$}{[n,4]}-\texorpdfstring{$\MDS(3)$}{MDS(3)}}

The goal of this section is the following theorem.

\begin{theorem}\label{thm:thm4}
There exists an explicit $\change{[}n,k\change{]}$-MDS code $V$ over field size
$O(n)^{2k-1}$ such that for any $A_1, A_2, A_3 \subseteq [n]$ such
that $|A_1| = 2$ and $|A_2| = |A_3| = k-1$ we have that $V_{A_1}
\cap V_{A_2} \cap V_{A_3} = 0$ whenever it generically should.
\end{theorem}

This result immediately implies a $\change{[}n,4\change{]}$-$\MDS(3)$ construction..

\begin{corollary}
There exists an explicit $\change{[}n,4\change{]}$-$\MDS(3)$ code over field size $O(n^7)$.
\end{corollary}

\begin{proof}
  By Corollary~\ref{cor:red}, the only nontrivial condition for
  $\change{[}n,4\change{]}$-$\MDS(3)$ is $|A_1| = 2$ and $|A_2| = |A_3| = 3$, so the
  construction applies.
\end{proof}

\begin{remark}
By similar logic, Theorem~\ref{thm:thm4} allows one to construct a
$\change{[}n,3\change{]}$-$\MDS(3)$ code over field size $O(n^5)$, although the
constructions in Section~\ref{sec:n3mds3} are superior. 
\end{remark}

\begin{proof}[Proof of Theorem~\ref{thm:thm4}.]

Let $q \ge n$ be a prime power with characteristic at least $k$. Let
$K$ be the a degree-$2k-1$ extension of $\F_q$. We let $\gamma$ be a
generator of this extension.

Pick an arbitrary injective map $\alpha : [n] \to \F_q$. We shall let
our RS field evaluations be
\[
  \beta_i = \gamma \alpha_i -\alpha_i^2.
\]

We let $V(\alpha, \gamma)$ be the dimension $k$ RS code
with $\beta_1, \hdots, \beta_n$ (as defined above) as evaluation points.

Let $A_1, A_2, A_3$ be such that $|A_1| = 2$ and
$|A_2| = |A_3| = k-1$, and $W_{A_1} \cap W_{A_2} \cap W_{A_3} = 0$ for
a generic matrix $W$. By Corollary~\ref{cor:weak}, to check that
$V_{A_1} \cap V_{A_2} \cap V_{A_3} = 0$, it suffices to check our
construction works for $A'_i \subseteq A_i$ in the code $V'$ (i.e.,
the $\change{[}n,k'\change{]}$-RS code with generators $\beta_1, \hdots, \beta_n$) where
$k' = \frac{1}{2}(|A'_1| + |A'_2| + |A'_3|)$ which are pairwise
disjoint.

If $|A'_2|$ or $|A'_3|$ is at least $k'$, then
$V_{A'_1} \cap V_{A'_2} \cap V_{A'_3}$ would be what it generically
should be based on $V$ being $\MDS$. Otherwise, $|A'_2| + |A'_3| \le
2k'-2$, so $|A'_1| \ge 2 = |A_1|$. Thus, we may assume without loss of generality that
$A'_1 = A_1 = \{1,2\}$ and that $|A'_2| = |A'_3| = k'-1$.
We may also assume $k' > |A_1| = 2$ by the same reasoning.

We seek to show that (\ref{eq:MDS-key}) has full row rank (i.e.,
nonzero determinant). That is, it suffices to consider the
determinant.
\[
  \Pi_{A'_2}(\beta_1)\Pi_{A'_3}(\beta_2) - \Pi_{A'_2}(\beta_2)\Pi_{A'_3}(\beta_1).
\]

Now, observe that for some nonzero $a \in \F_q$, we have that
\[
  \Pi_{A'_2}(\beta_1)\Pi_{A'_3}(\beta_2) = a \prod_{i \in A'_2} (\gamma -
  (\alpha_1 + \alpha_i)) \prod_{j \in A'_3} (\gamma -  (\alpha_2 + \alpha_j)).
\]
Likewise for some nonzero $b \in \F_q$,
\[
  \Pi_{A'_2}(\beta_2)\Pi_{A'_3}(\beta_1) = b \prod_{i \in A'_2} (\gamma - 
  (\alpha_2 + \alpha_i)) \prod_{j \in A'_3} (\gamma -  (\alpha_1 + \alpha_j)).
\]
Assume for sake of contradiction that
$\Pi_{A'_2}(\beta_1)\Pi_{A'_3}(\beta_2) = \Pi_{A'_2}(\beta_2)\Pi_{A'_3}(\beta_1)$.
Since $\gamma$ is a degree $n^{2k-1}$ extension, by comparing the lead
coefficients, we know that $a = b$, and the polynomials in $\gamma$
must have the same roots. Thus, there exists a permutation $\pi : (\{1\}
\times A'_2) \cup (\{2\} \times A'_3) \to (\{2\} \times A'_2) \cup (\{1\}\times A'_3)$ such that whenever $\pi(i_1,i_2) = (j_1,j_2)$ we have that
$\alpha_{i_1} + \alpha_{i_2} = \alpha_{j_1} + \alpha_{j_2}$. Since we
are assuming that $A'_2$ and $A'_3$ are disjoint, we can never map
$\{1\} \times A'_2$ to $\{1\} \times A'_3$ nor $\{2\} \times A'_3$ to
$\{2\} \times A'_2$. Thus, we can decompose $\pi$ into a pair of
permutations $\tau_2 : A'_2 \to A'_2$ and $\tau_3 : A'_3 \to A'_3$ such
that for all $i \in A'_2$ we have that $\alpha_1 + \alpha_i = \alpha_2
+ \alpha_{\tau_2(i)}$ with a similar definition for $\tau_3$. Now
observe that since $\F_q$ has characteristic at least $k$ we have that
\[
  0 \neq |A'_2|(\alpha_2 - \alpha_1) = \sum_{i \in A'_2} (\alpha_i -
  \alpha_{\tau_2(i)}) = \sum_{i \in A'_2} \alpha_i - \sum_{i \in A'_2}
  \alpha_i = 0,
\]
a contradiction.
\end{proof}

\subsection{\texorpdfstring{$\change{[}n,5\change{]}$}{[n,5]}-\texorpdfstring{$\MDS(3)$}{MDS(3)}}

We generalize the construction a bit in the weak-MDS case; that is,
where the sets considered are disjoint (c.f., \cite{bgm2021mds,roth2021higher}).

\begin{theorem}\label{thm:5}
  There exists an explicit $\change{[}n,k\change{]}$-MDS code $V$ over field size
  $O(n)^{2k^2}$ such that for any $A_1, A_2, A_3 \subseteq [n]$ such
  that $|A_1| + |A_2| = k+1$ and $|A_3|=k-1$ we have that
  $V_{A_1} \cap V_{A_2} \cap V_{A_3} = 0$ whenever it generically should.
\end{theorem}

Note that for any $k \le 5$, if $|A_1| + |A_2| + |A_3| = 2k$ and each
has size at most $k-1$, then at least one $|A_i| = k-1$. Thus, we have
the following immediately corollary.

\begin{corollary}
There exists an explicit $\change{[}n,5\change{]}$-$\MDS(3)$ code over field size $O(n^{50})$
\end{corollary}

\begin{proof}[Proof of Theorem~\ref{thm:5}.]
Pick distinct $\alpha_1, \hdots, \alpha_n \in \F_q$ such that
$\alpha_{i_1} + \alpha_{i_2} = \alpha_{j_1} + \alpha_{j_2}$ iff
$\{i_1,i_2\} = \{j_1,j_2\}$. It is clear that this works for
$q \approx n^2$. For example, over characteristic 2, we can take $\alpha_1,\dots,\alpha_n$ to be the columns of a parity check matrix of a BCH code with distance 5. Let $\F_q[x]$ be the univariate polynomial ring over
$\F_q$. We define our Reed-Solomon evaluation polynomials to be
$\beta_i(x) = \alpha_i x - \alpha_i^2$ for all $i \in [n]$. A key
observation is that due to our choice of $\alpha_i$'s, whenever
$i_1 \neq i_2$, $j_1 \neq j_2$, and $\{i_1,i_2\} \neq \{j_1,j_2\}$ we
have that\footnote{Here we assume that $\gcd$ always returns a monic
  polynomial.}
\begin{align}
\gcd(\beta_{i_1} - \beta_{i_2}, \beta_{j_1} - \beta_{j_2}) = 1.\label{eq:beta-gcd}
\end{align}

Let $V(x)$ be the
$\change{[}n,k\change{]}$-RS code generated by these evaluation points. We claim that
for any $A_1, A_2, A_3 \subseteq [n]$ with
$|A_1|+|A_2| = k+1$ and $|A_3| = k-1$ and $W_{A_1} \cap W_{A_2} \cap W_{A_3}$, we have that
\begin{align}
  V_{A_1}(x) \cap V_{A_2}(x) \cap V_{A_3}(x) = 0.\label{eq:V(x)}
\end{align}

Since this condition can be written as a degree $\le k^2$ polynomial
in $x$, by evaluating $V$ at some degree $k^2$ irreducible of $\F_q$,
we then immediately get an $n^{2k^2}$ construction.\footnote{We
  suspect this construction works over a smaller field extension.}

By Corollary~\ref{cor:weak}, it suffices to prove $V'_{A'_1}(x) \cap
V'_{A'_2}(x) \cap V'_{A'_3}(x) = 0$, where $V'$ is the $\change{[}n,k'\change{]}$-RS code with
the same generators where $k' = (|A'_1| + |A'_2| + |A'_3|)/2$, $A'_i
\subseteq A_i$, and $A'_1, A'_2, A'_3$ are disjoint. If $|A'_i| \ge
k'$ for some $i$, then the condition on $V'$ follows from $V'$ being
MDS. Since $|A'_3| \ge k'-1$ by Corollary~\ref{cor:weak}, we may assume that $|A'_1| = a'$,
$|A'_2| = k'+1-a'$ and $|A'_3| = k'-1$. %

Assume without loss of generality that $A'_1 = \{1,2,\hdots, a'\}$. By Lemma~\ref{lem:prod-mat}, we have that
(\ref{eq:V(x)}) holds iff
\[
\det \begin{pmatrix}
\Pi^{a'-1}_{A'_2}(\beta_1) & \Pi_{A'_3}(\beta_1)\\
\Pi^{a'-1}_{A'_2}(\beta_2) & \Pi_{A'_3}(\beta_2)\\
\vdots & \vdots\\
\Pi^{a'-1}_{A'_2}(\beta_{a'}) & \Pi_{A'_3}(\beta_{a'})
\end{pmatrix} \neq 0.
\]
We now use an infinite descent argument.  Assume for sake of
contradiction that the determinant equals zero. Then, there exists
$g_1, \hdots, g_{a'} \in \F_q[x]$ with $\gcd(g_1, \hdots, g_{a'}) = 1$ such
that for all $i \in [a']$,
\[ 
  \left(\sum_{j=1}^{a'-1} \beta_i^{j-1} g_{j}\right)\Pi_{A'_2}(\beta_i) + g_{a'}
  \Pi_{A'_3}(\beta_i) = 0.
\]
Note that by (\ref{eq:beta-gcd}), we have that
$\gcd(\Pi_{A'_2}(\beta_i), \Pi_{A'_3}(\beta_i)) = 1$ for all $i \in
[a]$. Thus, $\Pi_{A'_2}(\beta_i)$ divides $g_{a'}$ for all $i \in [a']$. Further,
by (\ref{eq:beta-gcd}), we have further that $\prod_{i=1}^{a
'}
\Pi_{A'_2}(\beta_i)$ divides $g_{a'}$. Thus, for some $h_{a'} \in \F_q[x]$
(possibly zero), we have that for all $i \in [a']$,
\[
  \sum_{j=1}^{a'-1} \beta_i^{j-1} g_{j} + h_{a'} \Pi_{A'_3}(\beta_i) \prod_{\substack{j=1\\j
      \neq i}}^{a'} \Pi_{A'_2}(\beta_j) = 0.
\]

In particular, we have that for all $i \in [a'-1]$,
\[
  \sum_{j=1}^{a'-1} \beta_i^{j-1} g_{j} \equiv 0 \mod \Pi_{A'_2}(\beta_{a'}).
\]

By standard properties of the Vandermonde determinant, we have that
for each $i \in [a'-1]$ there exists a polynomial $h_i \in \F_q[x]$
such that $g_i$ can be written as
$h_i \Pi_{A'_2}(\beta_{a'}) / \prod_{1 \le i < j \le a'-1} (\beta_j -
\beta_i)$.

Using (\ref{eq:beta-gcd}) one last time, we have that
$\prod_{1 \le i < j \le a'-1} (\beta_j - \beta_i)$ is relatively prime
to $\Pi_{A'_2}(\beta_{a'})$. Thus, each $g_i$ is divisible by
$\Pi_{A'_2}(\beta_{a'})$. This violates the assumption that $\gcd(g_1,
\hdots, g_{a'}) = 1$. Therefore, (\ref{eq:V(x)}) holds and our
construction is valid.
\end{proof}

\bibliographystyle{alpha}
\bibliography{clean-ref.bib}

\newcommand{\etalchar}[1]{$^{#1}$}
\begin{thebibliography}{GHK{\etalchar{+}}17}

\bibitem[ACKL23]{athi2023structure}
Harshithanjani Athi, Rasagna Chigullapally, Prasad Krishnan, and V.~Lalitha.
\newblock On the {{Structure}} of {{Higher Order MDS Codes}}.
\newblock In {\em 2023 {{IEEE International Symposium}} on {{Information
  Theory}} ({{ISIT}})}, pages 1009--1014, June 2023.

\bibitem[AGL24a]{alrabiah2023ag}
Omar Alrabiah, Venkatesan Guruswami, and Ray Li.
\newblock {{AG Codes Have No List-Decoding Friends}}: {{Approaching}} the
  {{Generalized Singleton Bound Requires Exponential Alphabets}}.
\newblock {\em IEEE Transactions on Information Theory}, 2024.

\bibitem[AGL24b]{alrabiah2023randomly}
Omar Alrabiah, Venkatesan Guruswami, and Ray Li.
\newblock Randomly {{Punctured Reed}}--{{Solomon Codes Achieve List-Decoding
  Capacity}} over {{Linear-Sized Fields}}.
\newblock In {\em Proceedings of the 56th {{Annual ACM Symposium}} on
  {{Theory}} of {{Computing}}}, pages 1458--1469, Vancouver BC Canada, June
  2024. ACM.

\bibitem[BDGZ24]{bdg2023b}
Joshua Brakensiek, Manik Dhar, Sivakanth Gopi, and Zihan Zhang.
\newblock {{AG Codes Achieve List Decoding Capacity}} over {{Constant-Sized
  Fields}}.
\newblock In {\em Proceedings of the 56th {{Annual ACM Symposium}} on
  {{Theory}} of {{Computing}}}, pages 740--751, Vancouver BC Canada, June 2024.
  ACM.

\bibitem[BEKS17]{bezanson2017julia}
Jeff Bezanson, Alan Edelman, Stefan Karpinski, and Viral~B. Shah.
\newblock Julia: {{A Fresh Approach}} to {{Numerical Computing}}.
\newblock {\em SIAM Review}, 59(1):65--98, January 2017.

\bibitem[BGM22]{bgm2021mds}
Joshua Brakensiek, Sivakanth Gopi, and Visu Makam.
\newblock Lower {{Bounds}} for {{Maximally Recoverable Tensor Codes}} and
  {{Higher Order MDS Codes}}.
\newblock {\em IEEE Transactions on Information Theory}, 68(11):7125--7140,
  November 2022.

\bibitem[BGM23]{bgm2022}
Joshua Brakensiek, Sivakanth Gopi, and Visu Makam.
\newblock Generic {{Reed-Solomon Codes Achieve List-Decoding Capacity}}.
\newblock In {\em Proceedings of the 55th {{Annual ACM Symposium}} on
  {{Theory}} of {{Computing}}}, {{STOC}} 2023, pages 1488--1501, New York, NY,
  USA, June 2023. Association for Computing Machinery.

\bibitem[CLO15]{cox2013ideals}
David~A. Cox, John~B. Little, and Donal O'Shea.
\newblock {\em Ideals, {{Varieties}}, and {{Algorithms}}: {{An Introduction}}
  to {{Computational Algebraic Geometry}} and {{Commutative Algebra}}}.
\newblock Undergraduate Texts in Mathematics. Springer, Cham, 4th edition
  edition, 2015.

\bibitem[DEF{\etalchar{+}}24]{OSCAR-book}
Wolfram Decker, Christian Eder, Claus Fieker, Max Horn, and Michael Joswig,
  editors.
\newblock {\em The {{Computer Algebra System OSCAR}}: {{Algorithms}} and
  {{Examples}}}, volume~32 of {\em Algorithms and {{Computation}} in
  {{Mathematics}}}.
\newblock Springer, 1 edition, August 2024.

\bibitem[DSY14]{dau2014gmmds}
Son~Hoang Dau, Wentu Song, and Chau Yuen.
\newblock On the {{Existence}} of {{MDS Codes}} over {{Small Fields}} with
  {{Constrained Generator Matrices}}.
\newblock In {\em 2014 {{IEEE International Symposium}} on {{Information
  Theory}}}, pages 1787--1791, June 2014.

\bibitem[GHK11]{Guruswami2010OnTL}
Venkatesan Guruswami, Johan H{\aa}stad, and Swastik Kopparty.
\newblock On the {{List-Decodability}} of {{Random Linear Codes}}.
\newblock {\em IEEE Transactions on Information Theory}, 57(2):718--725,
  February 2011.

\bibitem[GHK{\etalchar{+}}17]{gopalan2016}
Parikshit Gopalan, Guangda Hu, Swastik Kopparty, Shubhangi Saraf, Carol Wang,
  and Sergey Yekhanin.
\newblock Maximally {{Recoverable Codes}} for {{Grid-like Topologies}}.
\newblock In {\em Proceedings of the 2017 {{Annual ACM-SIAM Symposium}} on
  {{Discrete Algorithms}} ({{SODA}})}, Proceedings, pages 2092--2108. {Society
  for Industrial and Applied Mathematics}, January 2017.

\bibitem[GST24]{goldberg2021Singleton}
Eitan Goldberg, Chong Shangguan, and Itzhak Tamo.
\newblock Singleton-type {{Bounds}} for {{List-decoding}} and
  {{List-recovery}}, and {{Related Results}}.
\newblock {\em Journal of Combinatorial Theory, Series A}, 203:105835, April
  2024.

\bibitem[GZ23]{guo2023randomly}
Zeyu Guo and Zihan Zhang.
\newblock Randomly {{Punctured Reed-Solomon Codes Achieve}} the {{List Decoding
  Capacity}} over {{Polynomial-Size Alphabets}}.
\newblock In {\em 2023 {{IEEE}} 64th {{Annual Symposium}} on {{Foundations}} of
  {{Computer Science}} ({{FOCS}})}, pages 164--176, November 2023.

\bibitem[HSX{\etalchar{+}}12]{huang2012erasure}
Cheng Huang, Huseyin Simitci, Yikang Xu, Aaron Ogus, Brad Calder, Parikshit
  Gopalan, Jin Li, and Sergey Yekhanin.
\newblock Erasure {{Coding}} in {{Windows Azure Storage}}.
\newblock In {\em 2012 {{USENIX Annual Technical Conference}} ({{USENIX ATC}}
  12)}, pages 15--26, 2012.

\bibitem[KMG21]{kong2021new}
Xiangliang Kong, Jingxue Ma, and Gennian Ge.
\newblock New {{Bounds}} on the {{Field Size}} for {{Maximally Recoverable
  Codes Instantiating Grid-like Topologies}}.
\newblock {\em Journal of Algebraic Combinatorics}, 54(2):529--557, September
  2021.

\bibitem[Lov21]{lovett2018gmmds}
Shachar Lovett.
\newblock Sparse {{MDS Matrices}} over {{Small Fields}}: {{A Proof}} of the
  {{GM-MDS Conjecture}}.
\newblock {\em SIAM Journal on Computing}, 50(4):1248--1262, January 2021.

\bibitem[OSC24]{OSCAR}
Oscar -- {O}pen {S}ource {C}omputer {A}lgebra {R}esearch {S}ystem, {V}ersion
  1.2.0-{DEV}, 2024.

\bibitem[PGM13]{Plank2013ScreamingFG}
James~S. Plank, Kevin~M. Greenan, and Ethan~L. Miller.
\newblock Screaming {{Fast Galois Field Arithmetic Using Intel SIMD
  Instructions}}.
\newblock In {\em 11th {{USENIX}} Conference on File and Storage Technologies
  ({{FAST}} 13)}, pages 298--306, San Jose, CA, February 2013. USENIX
  Association.

\bibitem[Rot22]{roth2021higher}
Ron~M. Roth.
\newblock Higher-{{Order MDS Codes}}.
\newblock {\em IEEE Transactions on Information Theory}, 68(12):7798--7816,
  December 2022.

\bibitem[Sin64]{Singleton1964maximum}
Richard Singleton.
\newblock Maximum {{Distance Q-Nary Codes}}.
\newblock {\em IEEE Transactions on Information Theory}, 10(2):116--118, April
  1964.

\bibitem[ST20]{shangguan2020combinatorial}
Chong Shangguan and Itzhak Tamo.
\newblock Combinatorial {{List-Decoding}} of {{Reed-Solomon Codes}} beyond the
  {{Johnson Radius}}.
\newblock In {\em Proceedings of the 52nd Annual ACM SIGACT Symposium on Theory
  of Computing}, STOC 2020, page 538–551. Association for Computing
  Machinery, 2020.

\bibitem[ST23]{ST23}
Chong Shangguan and Itzhak Tamo.
\newblock Generalized {{Singleton Bound}} and {{List-Decoding Reed}}--{{Solomon
  Codes Beyond}} the {{Johnson Radius}}.
\newblock {\em SIAM Journal on Computing}, 52(3):684--717, June 2023.

\bibitem[Tia19]{tian2019formulas}
Yongge Tian.
\newblock Formulas for {{Calculating}} the {{Dimensions}} of the {{Sums}} and
  the {{Intersections}} of a {{Family}} of {{Linear Subspaces}} with
  {{Applications}}.
\newblock {\em Beitr{\"a}ge zur Algebra und Geometrie / Contributions to
  Algebra and Geometry}, 60(3):471--485, September 2019.

\bibitem[YH19]{yildiz2019gmmds}
Hikmet Yildiz and Babak Hassibi.
\newblock Optimum {{Linear Codes With Support-Constrained Generator Matrices
  Over Small Fields}}.
\newblock {\em IEEE Transactions on Information Theory}, 65(12):7868--7875,
  December 2019.

\end{thebibliography}

\end{document}